\documentclass[preprint,12pt]{elsarticle}

 \usepackage{graphicx}
 \usepackage{color}
\usepackage{booktabs}

\usepackage{enumitem}
\usepackage{amssymb}
\usepackage{algorithm}
\usepackage{algorithmic}
\usepackage{amsthm}

\journal{}

\begin{document}

\begin{frontmatter}


%% \title{Aumann (strong Nash) equilibrium detection\tnoteref{label1}}

\title{Computing Strong Nash Equilibria for Multiplayer Games}


\author{No\'{e}mi Gask\'{o},   Rodica Ioana Lung, D. Dumitrescu\footnote{-All authors have equal contributions}}

\address{Babe\c s-Bolyai University, Cluj-Napoca, Romania}
\ead{\{gaskonomi, ddumitr\}@cs.ubbcluj.ro, rodica.lung@econ.ubbcluj.ro}

\begin{abstract}
A new method for computing strong Nash equilibria in multiplayer games that uses the theoretical framework of generative relations combined with a stochastic search method is presented. Generative relations provide a mean to compare two strategy profiles and to assess their relative quality with respect to an equilibria type.  
The stochastic method used, called Aumann Crowding Based Differential Evolution (A-CrDE), uses a Differential Evolution algorithm that has been successfully used for numerical optimization problem. Numerical experiments illustrate the efficiency of the approach. 
\end{abstract}


\end{frontmatter}

\section{Introduction}

\newtheorem{definition}{Definition}
\newtheorem{remark}{Remark}
\newtheorem{example}{Example}
\newtheorem{proposition}{Proposition}

Strong Nash equilibrium (SNE) or Aumann equilibrium is one of the most appealing equilibrium concepts in non-cooperative game theory \cite{GTintroducere,aum2,gintis2000game}. Proposed by Aumann \cite{aum} as an alternative to the Nash equilibrium (NE), SNEs  take into account the fact that some of the players, although having no unilateral incentive to deviate, may benefit (sometimes substantially)  from forming alliances/coalitions with other players. While in a NE no player can improve its payoff by unilateral deviation, in a SNE there is no coalition of players that can improve their payoffs (by collective deviation). Thus, SNEs present the advantages of a cooperative behavior in a non-cooperative environment. 

Two major downsides appear when dealing with SNEs:
\begin{itemize}[noitemsep,nolistsep]
\item SNEs need not exist for all games; however, this paper is concerned only with games that present at least one SNE;
\item the computational complexity related to the necessity of considering all possible  coalitions among players.
\end{itemize}

In spite of that,  the strong Nash equilibrium  is a robust, worth exploring equilibrium concept; the importance of SNEs is widely studied for classes of games that allow the characterization of SNEs, such as congestion games \cite{Holzman199785}, network games \cite{Holzman2003193, Matsubayashi2006387}, voting models  \cite{Keiding2001117,Moulin1982}, etc. 

% The most popular equilibrium concept is the Nash equilibrium \cite{nash}, when no player can deviate in order to increase her/his payoff.  Among the major limitations of NE \cite{Farrell} lies the problem of the multiple Nash equilibria. If a game has several Nash equilibria the selection problem, in which players have no extra criteria to select one NE from the existing ones, can appear. To solve this problem NE refinements were introduced: Aumann (strong Nash) equilibrium \cite{aum}, coalition proof Nash equilibrium \cite{bernheim}, modified strong Nash equilibrium  \cite{greenberg,ray}, strong Berge equilibrium \cite{berge}, strong mediated equilibrium \cite{Monderer06strongmediated}, etc.

Although the existence and properties of SNEs have been studied \cite{Nessah2014871}, there are few computational tools available for computing the SNEs. The complexity of computing a SNE is known to be $\mathcal{NP}-complete$ \cite{Conitzer2008621, DBLP:conf/atal/GattiRS13}. For pure strategy SNEs there are several algorithms designed for specific classes of games: congestion games \cite{0899.90169,
Hayrapetyan:2006:ECC:1132516.1132529, Rozenfeld:2006:SCS:2081411.2081419,
Hoefer:2010:CPN:1929237.1929264}, connection games \cite{Epstein200951}, continuous games \cite{Nessah2014871}. An algorithm for detecting strong Nash equilibrium  in bottleneck congestion games is described in \cite{harks}. Properties, existence conditions, and an analytical algorithm are described in \cite{Nessah2014871}.

The aim of this article is to compute SNEs using a heuristic search algorithm.  In order to accomplish that, a method to compare two strategy profiles with respect to the characteristics of SNEs and decide if one is ``better'' than the other is needed. Such a binary relation has been proposed in \cite{gasko1} and successfully used to approximate SNEs. This paper  studies theoretical aspects related to this relation and furthermore proposes two variants that are less computational consuming in terms of running time.

The rest of the article is organized as follows: Section \ref{sec:theory} presents some basic Game Theory notions (non-cooperative game, Nash equilibrium, strong Nash equilibrium). Section \ref{sec:generativerelations} describes the generative relations necessary for the equilibrium detection (the strong Nash, probabilistic strong Nash non-dominated relation) and Section \ref{sec:evo} the evolutionary approach. In Section \ref{sec:numericalexperiments} numerical experiments are presented. The paper ends with Conclusions.

%---------------------------------------------------------

\section{Strong Nash equilibria definitions}
\label{sec:theory}
%Non-cooperative Game Theory can model independent decision making.

A non-cooperative game is described by a system of players, actions and payoffs. Each player has some actions/strategies available and a payoff function that takes into account the actions of all players. 

Formally, a finite strategic non-cooperative game $\Gamma$ is a system $\Gamma=(N,S,U),$ where:

\begin{itemize}[noitemsep,nolistsep]
\item $N$ represents a set of players, and $n$ is the number of players;

\item for each player $i \in N$, $S_{i}$ is the set of actions available, and  $$S=S_{1} \times S_{2} \times ... \times S_{n}$$
is the set of all possible situations of the game.
An element $s \in S$, $s=(s_1,s_2,...,s_n)$, is called a strategy (or strategy profile) of the game with $s_i$ denoting the strategy of player $i$;

\item $U=(u_1,...,u_n)$ is the set of payoff functions; for each  $i \in N$, $u_{i}:S \rightarrow R$ represents the payoff function of player $i$.
\end{itemize}

Let $\mathcal P(N)$ be the power set of $N$, containing all possible player coalitions and $I$ a nonempty set of $\mathcal P(N)$. Then $N-I=\{i\in N; i\not\in I\}$ is the set of the rest of the players. If $I=\{i\}$, i.e. contains only one player, instead of $N-I$ we will write $-i$. Using these notations, if $s,q\in S$, $(s_I, q_{N-I})$ denotes the strategy  in which players from $I$ play their strategies from $s$ and players from $N-I$ their strategies from $q$. If $I=\{i\}$, $(s_i, q_{-i})=(q_1,...,q_{i-1},s_i,q_{i+1},...,q_n)$.

The Nash  equilibrium  \cite{nash} is a strategy profile such that no player can unilaterally change her/his strategy to increase her/his payoff.

\begin{definition}[Nash equilibrium]
A strategy profile $s^{*} \in S$ is a Nash equilibrium if the inequality
$$u_i(s_{i}^{},s_{-i}^{*}) \leq u_i(s^{*}),$$
 holds,  $\forall i=1,...,n, \forall s_{i} \in S_i.$
\end{definition}

A Pareto efficient (or optimal) strategy is a situation in which no player can improve his/her payoff without decreasing the payoff of someone else. 
\begin{definition}[Pareto efficiency]
A strategy profile $s^* \in S$ is Pareto efficient if there does not exist a strategy $s \in S$ such that $$u_i(s)\geq u_i(s^*), \forall  i \in N,$$
with at least one strict inequality.
\end{definition}

%\subsection{Aumann (strong Nash) equilibrium}

The strong Nash (or Aumann) equilibrium is a strategy for which no coalition of players has a joint deviation that improves the payoff of each member of the coalition.

\begin{definition}[Strong Nash equilibrium]
The strategy
$s^{*}$ is a strong Nash (Aumann) equilibrium  if 
$\forall I \subseteq N, I \neq \emptyset$ there does not exist any $s_{I}$ such that
%there exists a player $i \in I$ such that 
the inequality
$$u_{i}(s_{I}^{},s^{*}_{N-I})> u_i(s^{*})$$
holds $\forall i\in I$.
\end{definition}

Let us denote by  $SNE(\Gamma)$ the set of strong Nash equilibria of the game $G$ and by  $NE(\Gamma)$ the set of Nash equilibria in the game $\Gamma$. 

The following remarks about SNEs are obvious from the definition.
\begin{remark}\label{rem:observatii}
\begin{itemize}[noitemsep]
\item Considering that if we choose deviating coalitions composed from a unique player it is clear that the strong Nash equilibrium reduces to the Nash equilibrium and we can write  $$SNE(\Gamma)\subseteq NE(\Gamma).$$ 
\item The definition of $SNE$ implies that any $SNE$ is Pareto efficient \cite{Nessah2013353}. Evenmore,  Nash equilibrium that is also Pareto efficient is a strong Nash equilibrium \cite{Gatti:2013:VCS:2484920.2485034}.
\item $SNE$ does not always exists in any non-cooperative games.
\end{itemize}
\end{remark}
\begin{example}\label{ex:ex1}
Let us consider a two  person coordination  game with payoffs presented in Table \ref{table:example1}.
The game has two NEs in pure form: $(A,A)$ and $(B,B)$, with the corresponding payoffs $(5,5)$ and $(4,4)$, and one NE in mixed form. Only the strategy profile $(A,A)$ is a strong Nash equilibrium. 
\begin{table}[h]
\caption{The payoff matrix for Example \ref{ex:ex1}}
\begin{center}
  \begin{tabular}{  l | c | c | c |}
  \multicolumn{4}{r}{Player 2} \\ \cline{2-4}
  & & A & B \\ \cline{2-4}
{Player 1}  & A & (5,5)  & (3,1)  \\  \cline{2-4}
          &  B & (2,3) & (4,4) \\ \cline{2-4}
  \end{tabular}
\end{center}
\label{table:example1}
\end{table}
\end{example}

%------------------------------------------------------------------------------

\section{Generative relations}
\label{sec:generativerelations}

\subparagraph{Generative relations} are used to characterize a certain equilibrium type by using the non-dominance concept. A binary relation $R$ is defined on $S$. If we have $sRq$, with $s,q\in S$, then we say that $s$ dominates $q$ with respect to relation $R$. Conversely, if, for some $s$,  $\nexists q$ such that $qRs$, we call $s$ non-dominated with respect to relation $R$. 

Relation $R$ is called $generative$ for an equilibrium type if the set of nondominated strategy profiles with respect to relation $R$ equals the set of equilibria. 

A generative relation  for the Nash equilibrium was introduced in \cite{iccc2008}. Other generative relations were defined for modified strong Nash and coalition proof Nash equilibrium in \cite{MSE_CPN}, and for strong Berge equilibrium in \cite{SB}.
In what follows a generative relation for SNEs is presented.

\subsection{Generative relation for strong Nash equilibrium}
\label{sec_aumann}
In what follows we will assume that the considered game presents at least one strong Nash equilibrium. 

A relative quality measure of two strategies with respect to strong Nash equilibrium can be defined as   \cite{rel_gen_aum}:
$$a(s^*,s)=card[i \in I, \emptyset \neq I \subseteq N, u_{i}(s_{I}^{},s^{*}_{N-I})> u_{i}(s^{*}), s_{i}^{} \neq s^{*}_{i}],$$
where $card[M]$ denotes the cardinality of the multiset $M$ (an element $i$ can appear several times in $M$ and each occurrence is counted in $card[M]$). Thus, $a(s^*,s)$ counts the total number of players that would benefit from collectively  switching their strategies from $s^*$ to $s$.

\begin{definition}
Let $s^*,s \in S$.
We say that strategy $s^*$ is better than strategy $s$ with respect to strong Nash equilibrium (or strong Nash dominates strategy $s$), and we write $s^* \prec_{A} s $ if the following inequality holds:  $$a(s^*,s)<a(s,s^*).$$
\end{definition}

Thus, strategy $s^*$ is better in strong Nash sense than a strategy $s$ if there are less players that would be able to increase their payoffs by entering in a coalition that switches strategies from $s^*$ to $s$ than vice-versa.

\begin{definition}
The strategy profile $s^{*} \in S$ is called strong Nash non-dominated (ANS) if there is no strategy $s \in S, s \neq s^{*}$ such that: $$s\prec_{A} s^{*}.$$
\end{definition}

Our assumption is that $ \prec_{A} $ is a generative relation  for  strong Nash equilibria, i.e. the set of non-dominated strategies with respect to $\prec_{A}$ is equal to the set of \textit{strong Nash equilibria} of the game. In order to prove that, we will use the following property.

\begin{proposition}
\label{aum0}
A strategy profile $s^{*} \in S$ is a strong Nash equilibrium  if and only if the equality $$a(s^{*},s)=0$$
holds for all $s \in S$.
\end{proposition}

\begin{proof}
\textit{(i)}
Let $s^{*} \in S$ be a SNE. Suppose there exists a strategy profile $s \in S$, such that $a(s^{*},s)=w$, $w >0$. Therefore there exists a set $I, I \subseteq N, I \neq \emptyset $, and $i \in I$, such that $$u_{i}(s_{I}^{},s^{*}_{N-I})> u_i(s^{*}), s_{I}^{}\neq s^{*}_{I}.$$ This contradicts the definition of SNE.

\textit{(ii)}
Let $s^{*} \in S$ be a strategy profile such that $$\forall s \in S, a(s^{*},s)=0.$$ This means that
$$u_{i}(s_{I}^{},s^{*}_{N-I}) \leq u_i(s^{*})$$
for all $I\subseteq N,$ $i \in I$, and for any strategy $s \in S.$
Therefore $s^{*}$ is strong Nash equilibrium.
\end{proof}

\begin{proposition}
\label{eq:aumann1}
All SNEs are strong Nash non-dominated solutions, i.e.
$$\textrm{SNE}\subseteq ANS.$$
\end{proposition}

\begin{proof}
Let $s^{*} \in SNE$. Suppose $s^{*}$ is strong Nash dominated. Therefore there exists a strategy profile $s \in S$ dominating $s^*$:
$$s \prec_{A} s^*.$$
From the definition of relation $\prec_{A}$, we have $$a(s,s^{*})<a(s^{*},s),$$
and from Proposition \ref{aum0}: $$a(s^{*},s)=0.$$
Therefore $$a(s,s^{*})<0.$$ But this is impossible as $a(s,s^{*})$ denotes the cardinality of a multiset.
\end{proof}

\begin{proposition}
\label{eq:aumann2}
All strong Nash non-dominated solutions are strong Nash equilibria, i.e.
$$ANS \subseteq SNE.$$

\begin{proof}
Let $s^{*}$ be an strong Nash non-dominated strategy profile. Suppose $s^{*}\not\in SNE$. Therefore there must exist (at least one) non-empty coalition $J,$ $j \in J$ and a strategy $s_{J} \in S$, such that

\begin{equation}
\label{eq:eqaumann}
u_{j}(s_{J}^{},s^{*}_{N-J})> u_{j}(s^{*}), \forall j \in J.
\end{equation}

We consider the coalition fixed, i.e. $J=\{j_{i_1},j_{i_2},...,j_{i_k}\}.$
Let us denote  $q=(s_{J}^{},s^{*}_{N-J}).$ Eq. (\ref{eq:eqaumann}) can be written as:
\begin{equation}
\label{eq:eqaumann1}
u_{j}(q)> u_{j}(s^{*}), \forall j \in J.
\end{equation}
We have $$a(s^*,q)=card[i \in I, \emptyset \neq I \subseteq N, u_{i}(q_{I}^{},s^{*}_{N-I})> u_{i}(s^{*}), q_{i}^{} \neq s^{*}_{i}].$$
which is equivalent with:
$$a(s^*,q)=card[i \in I, \emptyset \neq I \subseteq N, u_{i}(s_{j}^{},s^{*}_{N-j})> u_{i}(s^{*}), s_j^{} \neq s^{*}_{i}].$$

From (\ref{eq:eqaumann1}) it follows that
\begin{equation}
\label{eq:eqaumann2}
a(s^{*},q) > 0.
\end{equation}
On the other hand we have:
$$a(q,s^*)=card[i \in I, \emptyset \neq I \subseteq N, u_{i}(s^{*}_{I},q_{N-I}^{})> u_{i}(q), s^{*}_{i} \neq q_{i}^{}, \forall i \in I].$$
But $q_{i}^{} \neq s^{*}_{i}$ holds only for $i=j,$ $j \in J.$
Hence
$$a(q,s^*)=card[j, u_{j}(s^{*}_{J},q_{N-J}^{})> u_{j}(q)].$$
But
$$(s^{*}_{J},q_{N-J}^{})=s^*.$$
Thus $a(q,s^*)$ can be written as follows:
$$a(q,s^*)=card[j, u_{j}(s^{*})> u_{j}(q)].$$
From (\ref{eq:eqaumann1}) it results, that:
\begin{equation}
\label{eq:eqaumann3}
a(q,s^*)=0.
\end{equation}
From (\ref{eq:eqaumann2}) and (\ref{eq:eqaumann3}) we have:
$$a(q,s^{*})<a(s^{*},q),$$ which means that $q \prec_{A} s^{*}.$ The hypothesis that $s^{*}$ is non-dominated is thus contradicted.
\end{proof}
\end{proposition}

\begin{proposition}\label{prop:totala}
Relation $\prec_{A}$ is a generative relation for strong Nash equilibria, i.e. $$ SNE=ANS.$$
\begin{proof}
Follows directly from Proposition \ref{eq:aumann1} and Proposition \ref{eq:aumann2}.
\end{proof}
\end{proposition}
Although $\prec_A$ is a generative relation for SNEs, it presents the major disadvantage that up to $n2^{n}$ payoff function evaluations are necessary with the corresponding computation of all $2^n-1$ possible coalitions. Because this makes $\prec_A$ impractical from a computational point of view when dealing with a large number of players, two alternatives are proposed in what follows. 

\subsection{Probabilistic generative relation for strong Nash equilibrium}
\label{sec_aumann_prob}

In order to reduce the number of evaluations a probabilistic model  that only takes into account some randomly generated coalitions is proposed.  

In the case of $n$ players the total number of possible coalitions is $2^n-1$. Consider a percent $p$ and $\mathcal A_p\subset \mathcal P(N)$ be a set of nonempty subsets of $N$ (possible coalitions) such that $card \{\mathcal A_p \}=[p (2^n-1)]$, where $[\cdot]$ denotes the integer part. The relative quality measure of two strategy profiles $s$ and $s^*$ with respect to $\mathcal A_p$ is:
$$a_p(s^*,s)=card[i \in I, \emptyset \neq I \subseteq \mathcal A_p, u_{i}(s_{I}^{},s^{*}_{N-I})> u_{i}(s^{*}), s_{i}^{} \neq s^{*}_{i}],$$

\begin{definition}
Let $s^*,s \in S$.
Strategy $s^*$ is $\mathcal A_p$-better than strategy $s$ with respect to strong Nash equilibrium (or $s^*$ probabilistic strong Nash dominates strategy $s$), and we write $s^* \prec_{A_p} s $ if the following inequality holds:  $$a_p(s^*,s)<a_p(s,s^*).$$
\end{definition}
 
Obviously, if $p=100\%$ the $\prec_{A_p}$ is identical to $\prec_A$.
\begin{definition}\label{def:pans}
For a given $p\neq 0$, strategy $s^{*} \in S$ is called $p$-strong Nash non-dominated ($pANS$) if there does not exists any $s \in S, s \neq s^{*}$ such that: $$s\prec_{A_p} s^{*}$$ for any $\mathcal A_p \subset \mathcal P(N)$ with $card \{\mathcal A_p \}=[p(2^n-1)]\neq 0$.
\end{definition}

In the following we will show that $ \prec_{A_p} $ is also a generative relation  for strong Nash equilibria, i.e. the set of $p -$non-dominated strategies with respect to $\prec_{A_p}$ approximates the set of \textit{strong Nash equilibria} of the game. 

\begin{proposition}
 For any $p\neq 0$, $\prec_{A_p}$ is a generative relation for strong Nash equilibria, i.e. $SE=pANS$.
 \begin{proof}
  The first implication is obvious (all SNEs are $p$-strong Nash non-domi\-nated), with the proof analogous  with that of Prop. \ref{eq:aumann1}.

 For the second one, $pANS \subseteq SE$, it is enough,  based on Prop. \ref{prop:totala}, to show that $pANS\subseteq ANS$. 
 
 Consider $s\in pANS$ such that $s\not \in ANS$. If $s\not \in ANS$, there exists $q\in S$ such that $q\prec_A s$, a.i. $$a(q,s)<a(s,q).$$
 
 Let $\{I_k\}_{k=1,...,m_q}$ be the family of $m_q$ coalitions such that $$u_i(q_{I_k},s_{N-I_k})>u _i(s)$$ and denote by $m=[p(2^n-1)]$, $m\neq 0$. 
 
 If $m>m_q$ then construct a family $\mathcal A_p$ by including all $I_k$ and any other coalitions. If $m\leq m_q$ then construct $\mathcal A_p$ consisting of only coalitions $I_k$, such that at least one  for which the relation $$u_i(s_I, q_{N-I})>u_i(q) \quad \forall i\in I$$ is not satisfied is included. Such a coalition exists, otherwise $q$ would not strong Nash dominate $s$. 
 
 Then, for $\mathcal A_p$, we can write: $$a_p(q,s)<a_p(s,q)$$ contradicting the hypothesis that $s\in pANS$. 
 
 \end{proof}

\end{proposition}

The probabilistic relation defined above presents the advantage of requiring less payoff function evaluations than $\prec_A$. 

%\subsection{Combined Nash-Pareto relation}
%\label{aum_appr}
%
%An alternative to the previous generative relations described in this paper is to combine the Pareto dominance relation with the Nash generative relation \cite{iccc2008}, as SNEs are both Pareto weekly efficient and Nash equilibria. Consider the following indicators of relative quality for two strategy profiles $s$ and $q$:
%$$a_P(s,q)=card\{i\in N|(u_i(s)<u_i(q))\},$$
%$$a_N(s,q)=card\{i\in N|u_i(s)<u_i(q_i, s_{-i}), s_i\neq q_i\}, $$ and
%$$a_{NP}(s,q)=card\{i\in N|(u_i(s)<u_i(q)) || (u_i(s)<u_i(q_i, s_{-i}) ), s_i\neq q_i  \}.$$
%
%Using these notation, Algorithm \ref{alg:NP} was empirically designed to compare two strategy profiles. Algorithm \ref{alg:NP} returns $1$ if $s$ NP-dominates $q$, $-1$ if $q$ NP-dominates $s$, and $0$ if they are indifferent to each other with respect to this combined relation. 
%
%\begin{algorithm}
%\caption{Nash-Pareto relation }\label{alg:NP}
%\textbf{ compare strategies $s$ and $q$ } 
%\begin{algorithmic}
%\IF {$a_P(s,q)<0.1n$} \STATE return 1;
%\ELSIF {$a_P(q,s)<0.1n$} \STATE return -1;
%\ELSIF {$a_{NP}(s,q)<a_{NP}(q,s)$} \STATE return 1;
%\ELSIF {$a_{NP}(s,q)>a_{NP}(q,s)$} \STATE return -1;
%\ELSIF {$a_{N}(s,q)<a_{N}(q,s)$} \STATE return 1;
%\ELSIF {$a_{N}(s,q)>a_{N}(q,s)$} \STATE return -1;
%\ELSE \STATE return 0;
%\ENDIF
%\end{algorithmic}
%\end{algorithm}

%****************************************************************************************
 
\section{Evolutionary approach}
\label{sec:evo}

Games - in which players try to simultaneously maximize their own payoffs - are similar with  multi-objective optimization problems (MOPs) in many features.  Population based metaheuristics that can deal with MOPs and are capable of  finding of the Pareto optimal set can easily be adapted, by using an appropriate generative relation, to compute certain game equilibrium types. 

In this section a new evolutionary algorithm, based on the Crowding based Differential Evolution algorithm for multimodal optimization \cite{thomsen04}, called strong Nash Crowding Differential Evolution  Algorithm (A-CrDE) is presented. A-CrDE uses the generative relations defined in Section \ref{sec:generativerelations} to evolutionary compute the strong Nash equilibria of a game. 

\paragraph{A-CrDE population}

The individuals from  population $P$ represent strategy profiles of the game ($s=(s_1,s_2,...,s_n)$, $n$ is the number of players) that are randomly initialized in the first generation. 

\paragraph{Crowding Differential Evolution}

CrDE  extends the Differential Evolution (DE) algorithm  with a crowding scheme \cite{thomsen04}. CrDE is based on the conventional DE, the only modification is made regarding the individual (parent) being replaced. Usually, the parent producing the offspring is substituted, whereas in CrDE the offspring replaces the most similar individual among  the population if it Aumann dominates it. A \emph{DE/rand/1/exp} scheme is used, as described in Algorithm \ref{alg:crde}. 

\begin{algorithm}
\caption{CrDE - the \emph{DE/rand/1/exp} scheme }\label{alg:crde}
\textbf{ \textit{Create offspring} $O[l]$ from parent $P [l]$ } 
\begin{algorithmic}[1]

\STATE $O[l] = P [l]$ %// copy parent genotype to offspring
\STATE randomly select parents $P [i_1 ]$, $P [i_2 ]$, $P [i_3 ]$, where $i_1 \neq i_2 \neq i_3 \neq i$
\STATE $n = U (0, dim)$
\FOR{$j=0$; $j < dim \wedge U (0, 1) < pc$; $j=j+1$}
\STATE $O[l][n] = P [i_1 ][n] + F \ast (P [i_2 ][n] - P [i_3 ][n])$
\STATE $n = (n + 1) \bmod{dim}$
\ENDFOR
\end{algorithmic}
\end{algorithm}

While a final condition is not fulfilled (for example the current
number of fitness evaluations performed is below the maximum number of evaluations allowed), for each individual $l$ from the population, an offspring $O[l]$ is created using the scheme presented in Algorithm \ref{alg:crde}, where $U (0, x)$ is a uniformly distributed number between $0$ and $x$, $pc$ denotes the probability of crossover, $F$ is the scaling factor, and $dim$ is the number of problem parameters (problem dimensionality, number of players in this case).

\paragraph{The generative relation}

Within CrDE the offspring $O[l]$ replaces the most similar parent $P[j]$ if it is fitter. Otherwise, the parent survives and is passed on to the next generation (iteration of the algorithm). Euclidean distance, or any other similarity measure can be used. Within A-CrDE an offspring replaces the parent if it is better than it with respect to the strong Nash equilibrium. Three variants of A-CrDE are considered, each one using one of the generative relations presented in Section \ref{sec:generativerelations}.

\paragraph{Outline of A-CrDE} 

The A-CrDE algorithm is outlined in Algorithm \ref{alg:crdemare}. The output of the algorithm consists on the set of strong Nash nondominated solutions in the last iteration that approximates the set of strong Nash equilibria of the game.

\begin{algorithm}
\caption{A-CrDE  }\label{alg:crdemare}
\begin{algorithmic}[1]
\STATE Randomly generate initial population $P_0$ of strategies;%// copy parent genotype to offspring
\WHILE{(not termination condition)} 
\FOR{each $l=\{1,...,population\: size\}$} 
\STATE create offspring $O[l]$ from parent $l$;
\IF {$O[l]$ \underline{\textit{strong Nash dominates}} the most \underline{similar} parent $j$} 
\STATE $O[l]$ replaces parent $j$;
\ENDIF
\ENDFOR
\ENDWHILE
\end{algorithmic}
\end{algorithm}

%-------------------------------------------------------
\section{Numerical experiments}
\label{sec:numericalexperiments}
%The CrDE algorithm is modified by replacing the Pareto dominance with the generative relation $\prec_{A_k}$. The new version of the CrDE method for the $k$-Aumann detection, the kA-CrDE is obtained. 
The detection of strong Nash equilibria using A-CrDE is illustrated for two games that present a SNE. Reported results are averaged over ten runs. Parameter settings used for numerical experiments are presented in Table \ref{tab:par}. Each run, the algorithm reports the distance to the real value of the SNE of the game. 

All of the experiments were run on a computer with 3.07 GHz CPU and 12 GB main memory.

Two different variants of A-CrDE that use:
\begin{itemize}
\item the generative relation proposed in \ref{sec_aumann} - A-CrDE;
\item the probabilistic generative relation proposed in \ref{sec_aumann_prob} - $p$A-CrDE, with $p$ taking values 10\%, 20\%, 30\%, and 40\%;
\end{itemize}
are studied.

\begin{table}
\begin{center}
 \caption{Parameter settings for A-CrDE used for the numerical experiments}
 \label{tab:par}
{\newcommand{\mc}[3]{\multicolumn{#1}{#2}{#3}}
\begin{tabular}{lcccc}\hline
Parameter & 2 & 3 & 4 & 5 \\ \hline
Pop size & \mc{4}{c}{50}\\
Max no evaluations (coalitional evaluations) & \mc{4}{c}{$2\times 10^9$}\\
 scaling factor F & \mc{4}{c}{0.5}\\
Crossover rate & \mc{4}{c}{0.9} \\ \hline
\end{tabular}
}%
\end{center}
\end{table}

\subsection{Game 1} 

The following continuous two person game \cite{nessah_tian}:
$$u_1(s_1,s_2)=3s_1^{2}-s_2^2+4s_2,$$
$$u_2(s_1,s_2)=-s_1^{2}+s_1-2s_2, $$
where $$s_i\in[-1,1],i=1,2,$$
presents one strong Nash equilibrium, the strategy profile $(1,-1)$ \cite{nessah}, located on the Pareto frontier of the game.

The algorithm correctly computes the strong Nash equilibrium (with all of the three proposed generative relations, in all runs). Figure \ref{fig:game2_payoff} pre\-sents the strong Nash equilibrium detected by A-CrDE for Game 1.

%The probabilistic Aumann detection methods detects also correctly the Aumann equilibrium of the game.

\begin{figure}
%\begin{minipage}[b]{0.55\linewidth}
\centering
\includegraphics[scale=0.75]{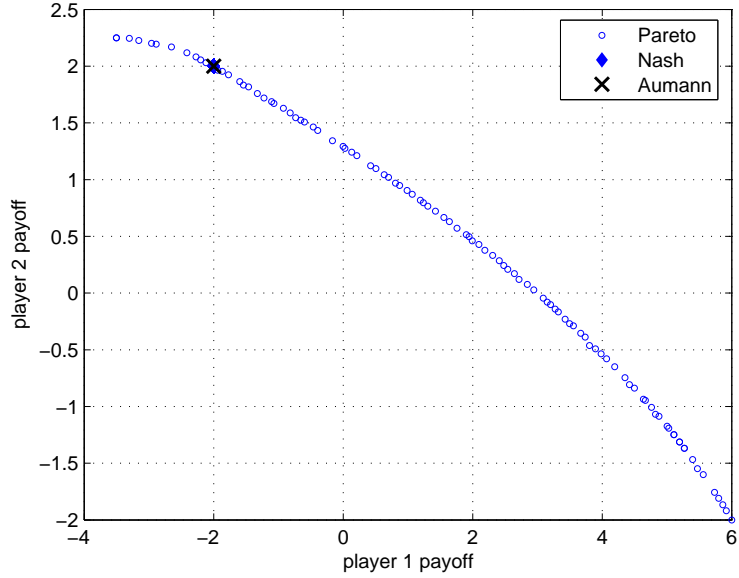}
\caption{Game 1. Strong Nash equilibrium detected by A-CrDE. Game 1 has only one NE that is also a SNE.}
\label{fig:game2_payoff}
%\end{minipage}
\end{figure}

%\subsection{Game 2}

%Let us consider Game 2 with the payoffs:

%$$u_1(s_1,s_2)=s_1(10-cos(s_1^2+s_2^2)),$$
%$$u_2(s_1,s_2)=s_2(10-cos(s_1^2+s_2^2)) s_i\in[0,10],i=1,2.$$

%The game has several Nash equilibria, from which a part is also Pareto efficient and also Aumann equilibrium. 

%\begin{table}
%\centering
%\begin{tabular}{cc}
% {Aumann} &{Aumann prob.}\\ \hline
% &   40\% \\ \hline
%  $0 \pm 0$ & $0 \pm 0$\\
%\end{tabular}
%\caption{Distance from the Aumann equilibrium and standard deviation for 10 different runs (Game $2$)}
%\end{table}

%\subsection{Game 2} 
%
%The following continuous two person game \cite{nessah_tian}:
%$$u_1(s_1,s_2)=-s_1^{2}-s_1+s_2,$$
%$$u_2(s_1,s_2)=2s_1^{2}+3s_1-s_2^2-3s_2, $$
%where $$s_i\in[0,1],i=1,2,$$
%presents one strong Nash equilibrium, the strategy profile $(0,0)$ \cite{nessah}, located on the Pareto frontier of the game.

%\subsection{Game 3}
%
%Let us consider Game 3 with the payoffs:
%
%$$u_1(s_1,s_2)=s_1(10-cos(s_1^2+s_2^2)),$$
%$$u_2(s_1,s_2)=s_2(10-cos(s_1^2+s_2^2)) s_i\in[0,10],i=1,2.$$
\subsection{Game 2: The minimum effort coordination game}

Game $G_2$ is based on a micro-foundation model \cite{bryant}. Another version of this game is presented in  \cite{Anderson2001177}. 

Consider a $n$-person  coordination game in which each player $i$ chooses an effort level $s_i$.  The common part of the effort is determined by the minimum effort of the $n$ effort levels. 
Each players' payoff is equal to the difference between the common payoff and the cost of the players own payoff ($\alpha s_i$):

$$u_i(s)=c_i-\alpha s_i,$$ where $$c_i=\min\{s_1,...,s_n\},$$

and $\alpha<1, s_i \in [0,10],i=1,...,n.$

Consider the cost $\alpha=0.5$.

The game has an infinite number of Nash equilibria (each $s_i=s$, $i=1,...,n$, $s,s_i \in [0,10]$ is a Nash equilibrium of the game), so all same effort level is a Nash equilibrium. The game has only one strong Nash equilibrium ($s_i=10$, $i=1,...,n$), which is Pareto efficient.

The objective space for the two player version is illustrated in \ref{fig:example2}.
\begin{figure}[H]\center
\includegraphics[width=0.6\linewidth]{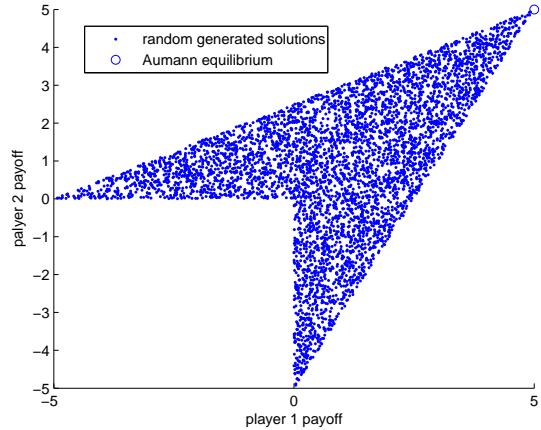}
\caption{Payoffs for the minimum effort coordination game for randomly generated strategies}
\label{fig:example2}
\end{figure}

Table \ref{table:sum} presents the distance to the strong Nash equilibrium and standard deviation for 10 different runs for A-CrDE and $p$A-CrDE. For 2, 5 and 10 players the algorithm finds correctly the strong Nash equilibrium. For 15 players a larger number of payoff function evaluations are necessary. 
\begin{table}
\centering
\small
\caption{Average distance to the strong Nash equilibrium and standard deviation for 10 independent runs}
\begin{tabular}{cccccc}\hline
{No. of pl.} &{A-CrDE} &\multicolumn{4}{c}{$p$A-CrDE}\\ \hline
 & & 10\% & 20\% & 30\% & 40\% \\ \hline
2 & $0 \pm 0$ & -  & - & - & $0 \pm 0$\\
5 & $0 \pm 0$ & $0 \pm 0$ &$0 \pm 0$ & $0 \pm 0$ & $0 \pm 0$\\
10 & $0 \pm 0$ & $0 \pm 0$ & $0 \pm 0$  & $0 \pm 0$&$0 \pm 0$\\
15 & $5.34 \pm 3.63$ &$2.76 \pm 8.30$ & $6.19 \pm 11.11$  & $ 8.84 \pm 14.52$  & $6.45 \pm 10.74$\\
\end{tabular}

\label{table:sum}
\end{table}

\begin{table}
\centering
\caption{Average run time (CPU seconds) for strong Nash equilibrium detection with different generative relations (10 independent runs)}
\small
\begin{tabular}{cccccc}\hline
{No. of pl.} &{A-CrDE} &\multicolumn{4}{c}{$p$A-CrDE}\\ \hline
 & & 10\% & 20\% & 30\% & 40\% \\ \hline
2 & 0.06 & -  & - & - & 0.06\\
5 & 0.1 &0.06 &0.06 &0.07& 0.08\\
10 & 5.82 & 0.73 & 1.25  & 2.05&2.53\\
15 & 328.66 &85.11 & 143.6  & 206.43  & 229.05\\
\end{tabular}

\label{table:time}
\end{table}

Table \ref{table:time} presents the average  time necessary to correctly compute the strong Nash equilibrium a single run (with distance 0). The results confirm that the algorithms are capable to locate the strong Nash equilibria even for 15 players, but they also indicate an exponential increase of the running time with the number of players.

To illustrate the evolution of the search, detailed results obtained for five players are depicted in Figures \ref{fig:game_convergence5} and \ref{fig:game_boxplot5}. Figure \ref{fig:game_convergence5} illustrates the convergence to the Aumann equilibrium for A-CrDE and $p$A-CrDE. Figure \ref{fig:game_boxplot5} depicts boxplots for the five cases, in each case the generation number is reported (for ten runs), in which the correct strong Nash equilibrium is obtained. The same type of results obtained for ten players are depicted Figures \ref{fig:game_convergence10} and \ref{fig:game_boxplot10}. The same conclusions can be drawn for both cases: in terms of payoff function evaluations, $p$A-CrDE seems to converge fastest for small values of $p$; A-CrDE has the slowest convergence. Boxplots representing the number of generations necessary to converge indicate that there are no significant differences between methods. A Wilcoxon sum rank test for assessing the statistical difference between means confirms this assumption.

\begin{figure}
\centering
\includegraphics[scale=0.65]{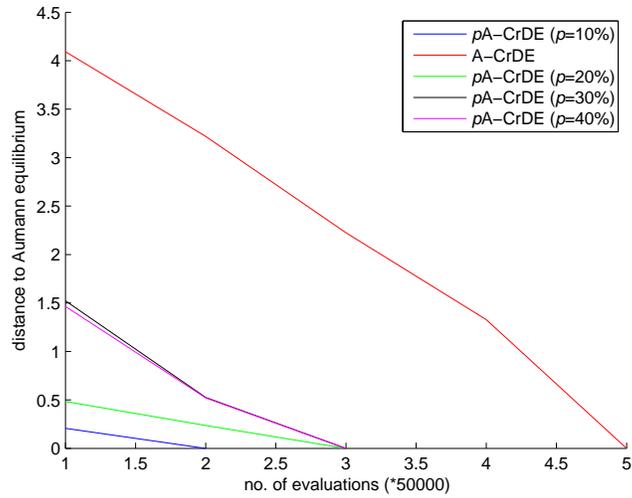}
\caption{Converge to the strong Nash equilibrium for 5 players. Values are smoothed using a moving average filter (with the MATLAB smooth function)}
\label{fig:game_convergence5}
\end{figure}

\begin{figure}
\centering
\includegraphics[scale=0.55]{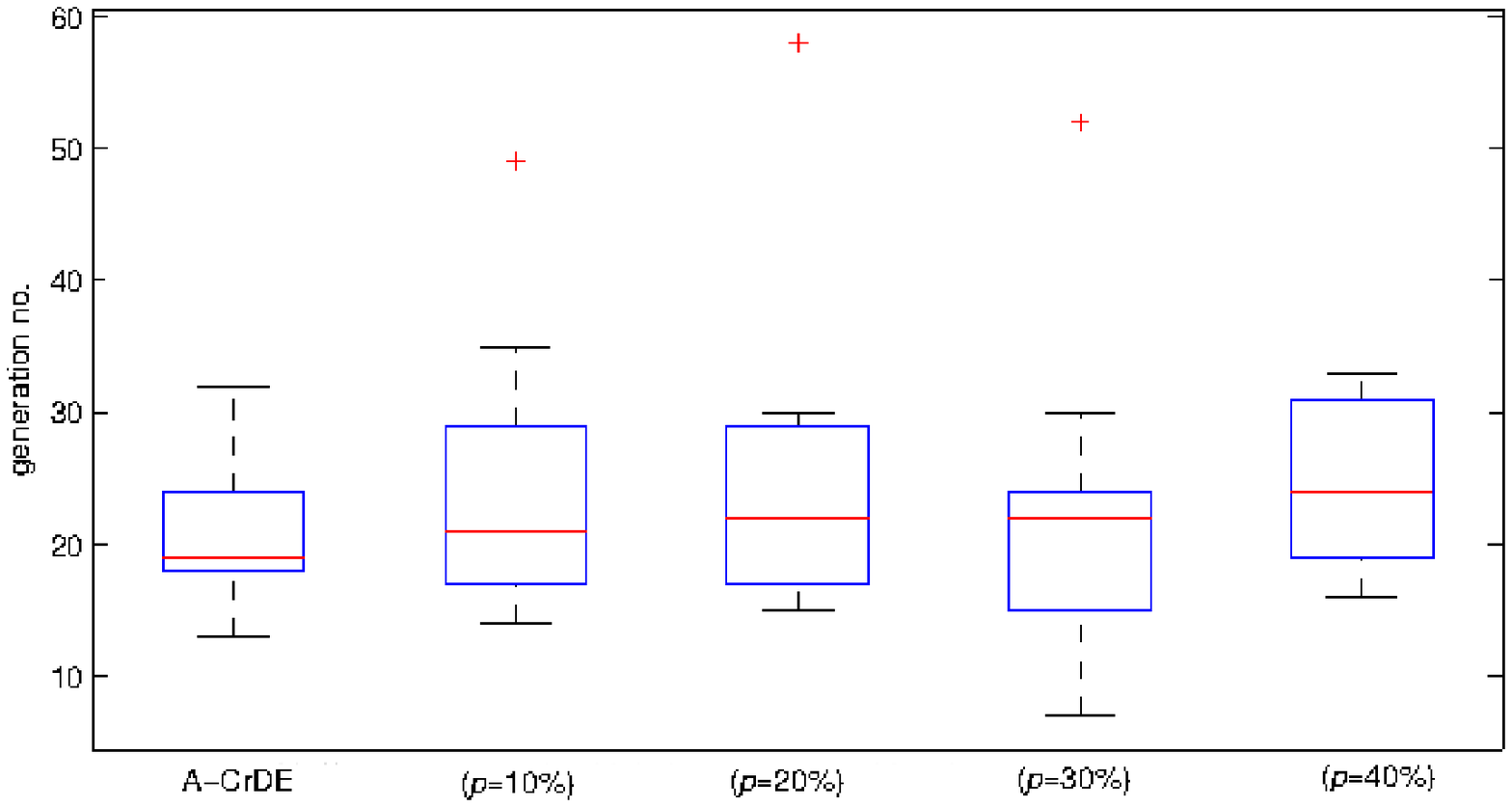}
\caption{Converge to the strong Nash equilibrium for five players (number of generations in 10 runs) }
\label{fig:game_boxplot5}
\end{figure}

\begin{figure}
\centering
\includegraphics[scale=0.65]{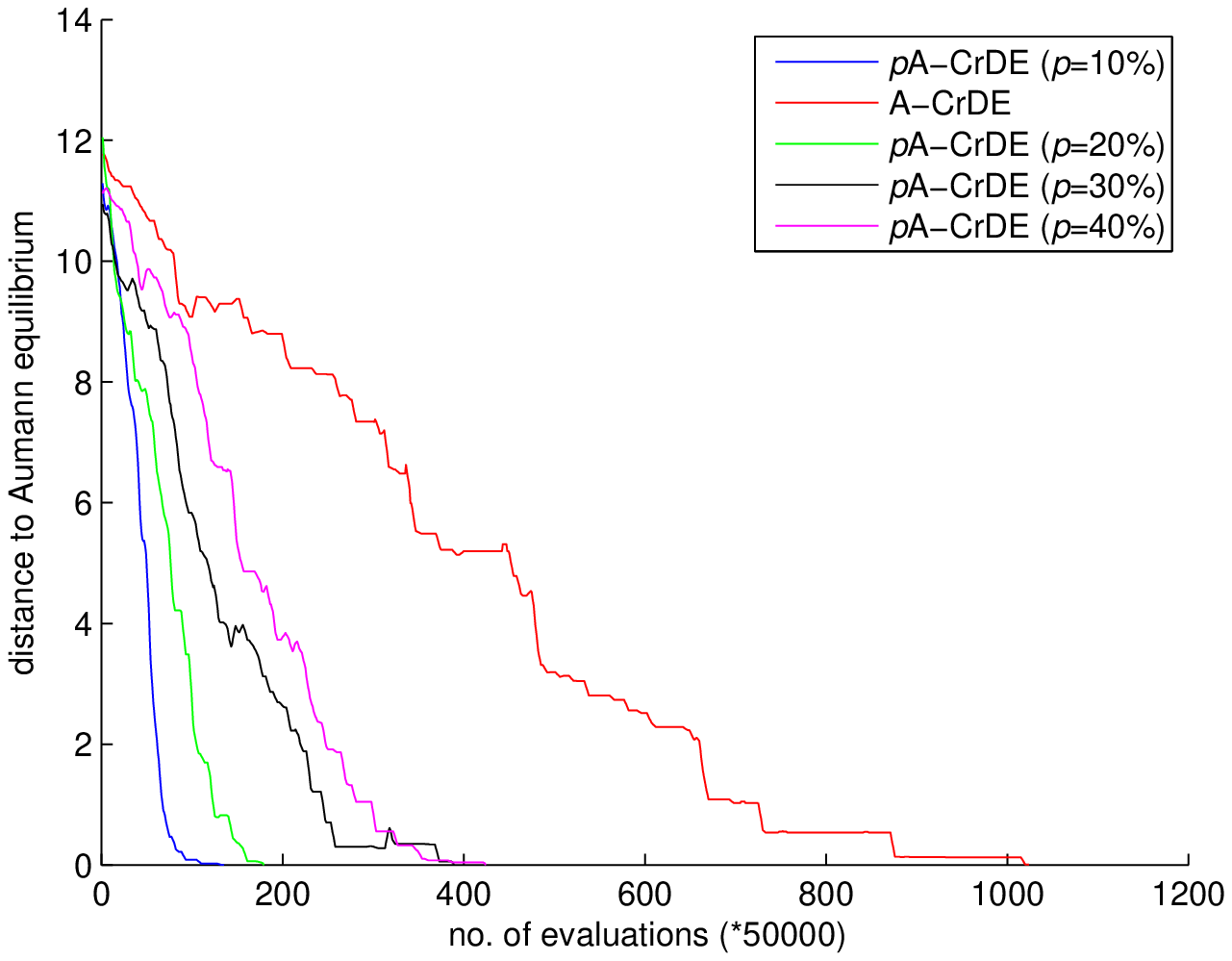}
\caption{Converge to the strong Nash equilibrium for ten players. The values are smoothed using a moving average filter (with the MATLAB smooth function)}
\label{fig:game_convergence10}
\end{figure}

\begin{figure}
\centering
\includegraphics[scale=0.55]{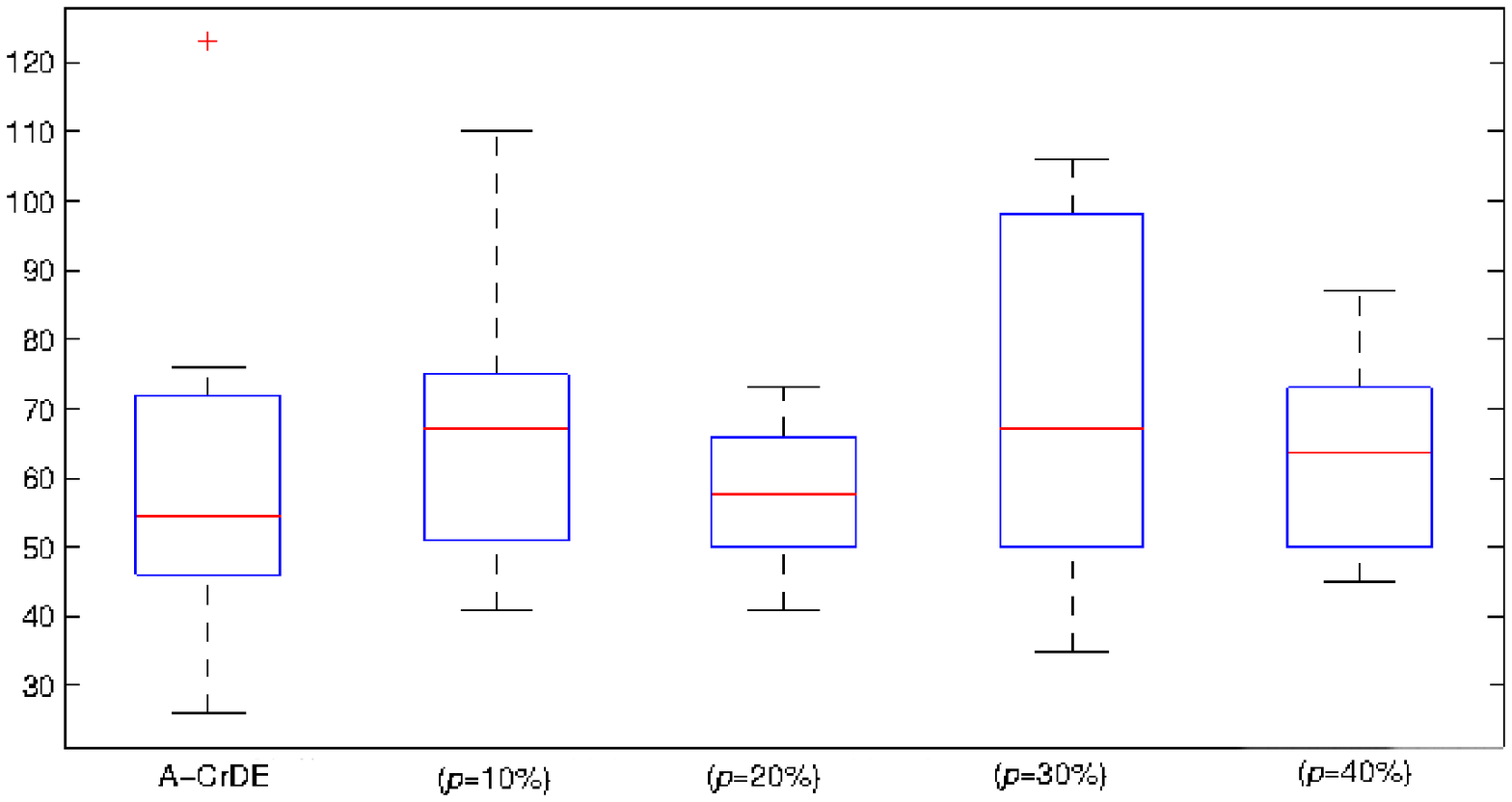}
\caption{Converge to the strong Nash equilibrium for ten players (number of generations in 10 runs) }
\label{fig:game_boxplot10}
\end{figure}

%The NPA-CrDe has a completely different behavior. Table \ref{table_aumann_appr} presents the results obtained with NPA-CrDE in terms of distance to the Aumann equilibrium and computational time. As it can be noticed, for game settings up to 150 players, NPA-CrDE correctly computes the SNE in reasonable time and significantly faster than A-CrDE and $p$A-CrDE. These results indicate the computational potential of NPA-CrDE for Aumann equilibria detection.
%
%\begin{table}
%\centering
%\caption{NPA-CrDE - Average distance to the Aumann equilibrium, standard deviation, and running time }
%\begin{tabular}{ccc}\hline
%{No. of pl.} &{Approximation} & Average computation time (sec)\\ \hline
%2 & $0 \pm 0$ & 0.04\\
%5 & $0 \pm 0$ & 0.04\\
%10 & $0 \pm 0$&0.09\\
%15 & $0 \pm 0$&0.3 \\
%25 & $0 \pm 0$&1.77 \\
%50 & $0 \pm 0$&20.79 \\
%75 & $0 \pm 0$& 117.40\\
%100 & $0 \pm 0$& 223.92\\
%125 & $0 \pm 0$&444.32 \\
%150 & $0 \pm 0$&712.22 \\
%\end{tabular}
%
%\label{table_aumann_appr}
%\end{table}

%-------------------------------------------------------
\section{Conclusions}

This paper presents an efficient approach to the problem of computing strong Nash equilibria by using evolutionary computation and generative relations. A theoretical framework presenting two generative relations and an empirically designed sets the basis for the evolutionary method. 

A differential evolution algorithm is adapted to search for SNEs by simply adding a generative relation in the replacing procedure of the method. 

Numerical examples illustrate the efficiency of the approach. For the minimum effort game, the third variant of the method correctly computes the SNE for instances up to 150 players. 

% Aumann equilibrium is a strong equilibrium concept in non-cooperative game theory. An evolutionary method is designed in order to detect the Aumann equilibrium by using an Aumann generative relation and a Differential Evolution method.
% 
% The main idea is that a multi-objective optimization problem is similar to an equilibrium detection problem. In both cases we have a maximization problem. For this kind of problems an evolutionary approach is a good choice. 
% 
% For the evolutionary approach an Aumann non-domination concept is introduced. The set of Aumman non-dominated solutions equals the set of the Aumann equilibria of the non-cooperative game. 
% 
% Three continuous games are described for numerical experiments. Results indicate the importance of the proposed method.

%---------------------------------------------------------


%% References


\bibliographystyle{elsarticle-num-names}

\bibliography{Nomibibliography,bib}

\end{document}